\documentclass[pra,aps,superscriptaddress,twocolumn,showpacs,nofootinbib]{revtex4-1}

\usepackage{graphicx,color}
\usepackage{amsmath,amsfonts,enumerate,amsthm,amssymb,bbm}
\usepackage{hyperref}
\usepackage{multirow}
\usepackage{bbold}
\usepackage{braket}

\hypersetup{
   colorlinks=true,         
   linkcolor=blue,          
   citecolor=blue,          
   urlcolor=blue          
}

\theoremstyle{plain}
\newtheorem{thm}{Theorem}

\newtheorem{lem}[thm]{Result}

\theoremstyle{definition}
\newtheorem{defn}{Definition}
\theoremstyle{remark}

\def\<{\langle}
\def\E{ {\cal E} }

\def\I{ \mathbb{1} }

\def\I{ \mathbbm{1} }
\newcommand{\tr}[2]{\mathrm{Tr}_{#1}\left[ #2 \right]}
\def\>{\rangle}
\def\<{\langle}

\renewcommand{\v}[1]{\ensuremath{\boldsymbol #1}}

\newcommand{\be}{\begin{equation}}
\newcommand{\ee}{\end{equation}}

\newcommand{\C}{\mathbb{C}}
\newcommand{\ketbra}[2]{\ensuremath{\left|#1\right\rangle\!\!\left\langle#2\right|}}

\begin{document}

\title{Thermodynamic resource theories, non-commutativity and maximum entropy principles}

\author{Matteo Lostaglio}
\affiliation{Controlled Quantum Dynamics Theory Group, Imperial College London, Prince Consort Road, London SW7 2BW, UK}
\affiliation{ICFO--Institut de Ciencies Fotoniques, The Barcelona Institute of Science and Technology, 08860 Castelldefels, Spain}
\author{David Jennings}
\affiliation{Controlled Quantum Dynamics Theory Group, Imperial College London, Prince Consort Road, London SW7 2BW, UK}
\author{Terry Rudolph}
\affiliation{Controlled Quantum Dynamics Theory Group, Imperial College London, Prince Consort Road, London SW7 2BW, UK}

\date{\today}

\begin{abstract}
	We discuss some features of thermodynamics in the presence of multiple conserved quantities. We prove a generalisation of Landauer principle illustrating tradeoffs between the erasure costs paid in different ``currencies''. We then show how the maximum entropy and complete passivity approaches give different answers in the presence of multiple observables. We discuss how this seems to prevent current resource theories from fully capturing thermodynamic aspects of non-commutativity. 
	
\end{abstract}


\maketitle

\section{Introduction and summary of results}

In this work we present several observations concerning the thermodynamics of systems with multiple and generally non-commuting conserved quantities. Our main results can be summarised as follow:
\begin{enumerate}
	\item In the first part of the paper, we prove a generalisation of Landauer erasure in the presence of multiple conserved charges. This stands in contrast to the standard assertion that erasure of information has an unavoidable \emph{energy} cost; we present simple tradeoffs among the different costs, e.g. energy and angular momentum. These are explicitly illustrated in the qubit case, where we give a tight protocol for information erasure using multiple baths.
	\item In the second part of the paper we discuss how, in the presence of multiple conserved quantities, different approaches to equilibrium can disagree. This gives a broader perspective on the tradeoffs analysed in the context of Landauer erasure.
	\item In the last part of the paper, we discuss the consequences of the previous results for the research program that looks at thermodynamics from a resource theory perspective. In particular, we argue that current resource-theoretic approaches are limited when it comes to determining the thermodynamic impact of non-commutativity.
\end{enumerate}

\section{Landauer erasure in the presence of multiple conserved quantities}

In his standard textbook \cite{callen1985thermodynamics}, Callen discusses the foundations of thermodynamics and asks: does energy play a unique role into it? Linking the first law of thermodynamics to Noether's theorem, he puts forward a natural question: ``Should not momentum and angular momentum play parallel roles with the energy?''. After all, they are generators of other fundamental symmetries of the physical world. He argues that they do; even more, he concludes ``The asymmetry in our account of thermostatistics is a purely conventional one that obscures the true nature of the subject''.

Here we show that this standpoint leads to a reconsideration of the meaning of Landauer's principle, traditionally stated as a fundamental thermodynamic relation imposing a minimum \emph{energy} cost $kT \log 2$ for the erasure of a bit of information in the presence of a heat bath at temperature $T$ ($k$ is Boltzmann's constant) \cite{plenio2001physics}. We consider situations involving multiple conserved quantities (such as energy and angular momentum) and provide explicit protocols illustrating the tradeoff between the costs of erasure in different charges.

\subsection{General bound}

Our derivation \cite{lostagliothesis} of the generalised Landauer's bound is based on the following assumptions, that follow the framework introduced in \cite{reeb2013proving}: 

\begin{enumerate}
	\item \label{assumption1} The initial state $S$ of the system, i.e.  the memory to be erased, and the reservoir $R$ are initially in a product state $\rho_{SR}=\rho_S \otimes \rho_R$. This is the most natural framework, because it models what happens in a typical erasure process. Indeed, allowing initial correlations between memory and reservoirs implies that we could erase a memory while extracting work at the same time, by exploiting the work value of correlations \cite{oppenheim2002approach, rio2010thermodynamic, llobet2015extractable}.
	\item \label{assumption2} The reservoirs have the form of a \emph{Generalised Gibbs Ensemble},
	\be
	\label{eq:firstgeneralisedgibbs}
	\rho_R = \frac{e^{-\sum_i \mu_i C_i }}{\tr{}{e^{-\sum_i \mu_i C_i }}},
	\ee
	for observables $C_i$ (we set $C_0:=H$, the Hamiltonian of the system and $\mu_0:=\beta = 1/(kT)$ the inverse temperature). This bath may factorize in the product of baths if all $C_i$ commute, \emph{but this is not necessary for the following}. 
	\item \label{assumption3} The total system of memory and reservoir is isolated, so it undergoes a general unitary evolution~$U$:
	\be
	\nonumber
	\rho_{SR} \stackrel{U}{\longmapsto} \rho'_{SR}.
	\ee
\end{enumerate}
To fix the notation,
\be
\nonumber
\rho'_S := \tr{R}{\rho'_{SR}}, \qquad \rho'_R := \tr{S}{\rho'_{SR}},
\ee
\be
\nonumber
\Delta C_i = \tr{}{C_i(\rho'_R - \rho_R)}, \quad S(\rho||\sigma) := - S(\rho) - \tr{}{\rho \log \sigma},
\ee
with $S(\rho) = -\tr{}{\rho \log \rho}$ the von Neumann entropy.  From the notation above, $\Delta C_0 := \Delta H$ is the heat flow towards the bath. Then,
\begin{lem}[Landauer principle for multiple conserved quantities]
	\label{lem:landauer}
Under the assumptions \ref{assumption1}, \ref{assumption2} and \ref{assumption3},
\be
\label{eq:landauerbound}
\beta \Delta H + \sum_{i\geq 1} \mu_i \Delta C_i \geq - \Delta S_S,
\ee
where $\Delta S_S = S(\rho'_S) - S(\rho_S)$.
\end{lem}
\begin{proof}
The proof of \cite{reeb2013proving} (Theorem~3) goes through independently of the non-commutativity of the $C_i$. Let \mbox{$\Delta S_X = S(\rho'_X) - S(\rho_X)$}, with $X=S,R$, and denote by \mbox{$I(S':R') := S(\rho'_S) + S(\rho'_R) - S(\rho'_{SR})$} the mutual information in the final state. Unitary invariance of the von Neumann entropy gives $S(\rho_{SR}) = S(\rho'_{SR})$. Then, 
\be
\label{eq:mutualinfo}
\Delta S_S + \Delta S_R = I(S':R'),
\ee 
i.e. the sum of the changes of the local entropies equals the correlations created in the transformation, as measured by the mutual information (this is a refinement of $\Delta S_S + \Delta S_R \geq 0$). Substituting Eq.~\eqref{eq:firstgeneralisedgibbs} in the expression for $\Delta S_R$ one finds
\be
\Delta S_R = \beta \Delta H + \sum_{i\geq 1} \mu_i \Delta C_i - S(\rho'_R||\rho_R).
\ee
This equation reduces to Clausius relation $\Delta H = k T \Delta S_R$ when we can make the assumption \mbox{$S(\rho'_R||\rho_R) \approx 0$} (which is typically the case for a macroscopic bath) and when only energy flows are involved. Together with Eq.~\eqref{eq:mutualinfo} this gives
\begin{small}
\be
\nonumber
-\Delta S_S + I(S':R') = \beta \Delta H + \sum_{i\geq 1} \mu_i \Delta C_i - S(\rho'_R||\rho_R).
\ee
\end{small}
From the non-negativity of mutual information and relative entropy, we obtain Eq.~\eqref{eq:landauerbound}.
\end{proof}

The above analysis indicates that the resultant trade-off is not particularly sensitive to the fact that the $C_i$ may not commute. The only impact the non-commutativity has is that throughout the sequence of bath interactions one obtains a ``time-dependent'' pairs of expectations $(\<C_i(t)\>, \<C_j(t)\>)$ subject to the uncertainty in the two observables.

\subsection{An explicit, tight protocol for qubit erasure}
\label{sec:qubitprotocol}

The simplest system with multiple conserved quantities is one with Hamiltonian $H$ and a single conserved charge $Q$, e.g, angular momentum in some fixed direction. As discussed the reservoir is assumed to have the form of Eq.~\eqref{eq:firstgeneralisedgibbs},
\begin{equation}
\label{eq:twochargesgibbs}
\gamma_R = \frac{e^{ - \beta H_R-\alpha Q_R}}{\tr{}{e^{ - \beta H_R-\alpha Q_R}}}.
\end{equation}
We take here $[H, Q]= [H_R, Q_R]=0$. $H_R$, $Q_R$ are observables on the bath Hilbert space physically corresponding to $H$ and $Q$ on the system.
$\beta$ and $\alpha$ are fixed inverse temperatures. Since we assume $[H_R,Q_R]=0$, one can formally factor the quantum state of Eq.~\eqref{eq:twochargesgibbs}. We will view it as if we have two independent baths at our disposal (a thermal bath and a ``$Q$-bath'') and we are able to put the system in contact with each of them separately.

Consider for simplicity the system Hilbert space \mbox{$\mathcal{H}=\C^2 \otimes \C^2$}, spanned by the eigenstates of two commuting observables $H \otimes \I $ and $\I \otimes Q$. Define the states $\{ \ket{00}, \ket{01}, \ket{10}, \ket{11} \}$, where $\ket{hq}:=\ket{h} \otimes \ket{q}$ and $\ket{h}$, $\ket{q}$ are eigenstates of  $H$ and $Q$, respectively. Assume these four states to be initially degenerate in energy $H$ and charge $Q$. Landauer's principle is usually stated as a fundamental lower bound on the \emph{energy} cost of the process of information erasure \cite{plenio2001physics}. We now show explicitly that this is not necessarily the case.

One can erase a qubit system in the angular momentum degree of freedom using a spin bath, with the erasure cost respecting a bound of the form $\Delta Q \geq \alpha^{-1} \log 2$ and zero energy cost \cite{vaccaro2011information} (here $\alpha$ is the inverse ``temperature'' of the spin bath). Going beyond this, we illustrate here the tradeoff between costs in different charges proven in Eq.~\eqref{eq:landauerbound}.

Let us encode classical bits in the states $\ket{00}$ and $\ket{10}$ in $\mathcal{H}$. 
We now develop a one-parameter family of optimal protocols (see Fig.~\ref{experiment}):

\begin{enumerate}
	\item Let us start with the levels $\ket{00}$ and $\ket{01}$ in a maximally mixed state (which describes the single bit memory to be erased):
	\be
	\nonumber
	\rho_{1} = 1/2\ketbra{00}{00} + 1/2\ketbra{10}{10}.
	\ee
	Also notice that, since $\ket{00}$ and $\ket{10}$ are degenerate in energy, $\rho_{1}$ is initially in thermal equilibrium with the thermal bath.
	
	\item Changing the Hamiltonian of the system, map the state $\ket{01}$ with energy $\bar{\epsilon}$ to a new state (still denoted by $\ket{01}$) with energy $\bar{\epsilon} + \Delta \bar{\epsilon}$, at an average cost $p(\bar{\epsilon}) \Delta \bar{\epsilon}$. Then put the system in contact to the heat bath and completely thermalise it with respect to the current Hamiltonian. Repeating these operations in a sequence of $N \rightarrow \infty$ steps with $\Delta \bar{\epsilon} \rightarrow 0$, the cost of raising the energy level from $0$ to $\epsilon$ is given by
	\be
	\nonumber
	\Delta H = \int_{0}^{\epsilon} \frac{e^{-\beta \bar{\epsilon}}}{1+ e^{-\beta \bar{\epsilon}}}d\bar{\epsilon} = \frac{1}{\beta}\ln \left( \frac{2}{1+e^{-\beta \epsilon}} \right).
	\ee
	The final state of the memory to be erased is 
	\be
	\nonumber
	\rho_{2} = \frac{1}{1+e^{-\beta \epsilon}} \ketbra{00}{00} + \frac{e^{-\beta \epsilon}}{1+e^{-\beta \epsilon}} \ketbra{10}{10}. 
	\ee
	\item Now we make use of the charge degree of freedom. The two states $\ket{00}$ and $\ket{01}$ are degenerate in charge $Q$. Following the same procedure as before we induce a level-raising $\ket{01}$ from $0$ to $q$, where we choose
\be
	\label{qchoice}
	q= \beta \epsilon / \alpha.
	\ee
	The reason for this choice will be clear in a moment. This can be done at no cost, because there is no population in $\ket{01}$, i.e. $p(q)=0$ throughout the process.\footnote{Also notice that we are assuming that the levels associated to the eigenstates of $Q$ are perfectly controllable. We will comment on this later in this section.} 
	
	\item We apply a unitary that performs the swap $\ket{00} \rightarrow \ket{00}$, $\ket{01} \leftrightarrow \ket{10}$, $\ket{11} \rightarrow \ket{11}$.
	The costs associated to this unitary are
	\be
	\nonumber
	\Delta H = -\frac{e^{-\beta \epsilon}}{1+e^{-\beta \epsilon}} \epsilon, \quad \Delta Q = \frac{e^{-\beta \epsilon}}{1+e^{-\beta \epsilon}} q,
	\ee
	and the final state, using Eq. \eqref{qchoice}, is
	\be
	\nonumber
	\rho_{3} = \frac{1}{1+e^{-\alpha q}} \ketbra{00}{00} + \frac{e^{-\alpha q}}{1+e^{-\alpha q}} \ketbra{01}{01}.
	\ee
	
	\item We can complete the erasure leaving the system in contact with the Q-bath and raising the level $\ket{01}$ from $q$ to $\infty$. Thanks to the choice of the initial $q$ (Eq. \eqref{qchoice}), $\rho_{3}$ is initially in equilibrium with the $Q$-bath. Hence this part of the protocol is formally analogous to steps 1-2, but is achieved using a physically different bath (e.g. a spin bath). The charge cost of the partial erasure in the charge basis is
	\be
	\nonumber
	\Delta Q =  \int_{\frac{\beta \epsilon}{\alpha}}^{+\infty} \frac{e^{-\alpha \bar{q}}}{1+ e^{-\alpha \bar{q}}}d\bar{q} = \frac{1}{\alpha}\ln (1+ e^{-\beta \epsilon}).
	\ee
	The final (erased) state of the memory is $\ket{00}$. This completes the protocol.
\end{enumerate}

\begin{figure}[t!]
	\includegraphics[width=5.5cm]{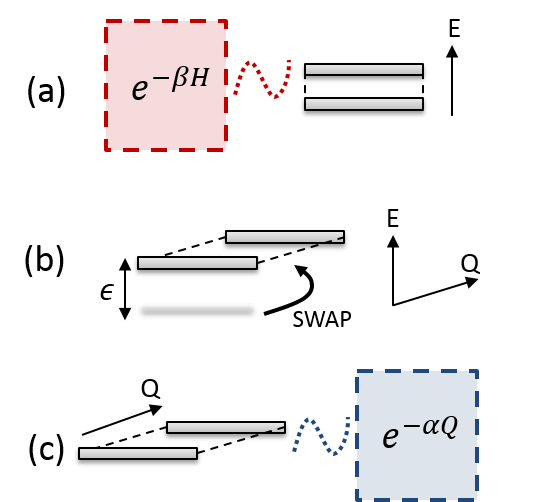}
	\caption{\label{experiment} \textbf{Generalised Landauer erasure.} A maximally mixed state on a two-level system is erased using both a thermal bath (a) and a spin bath (c). A unitary process interchanges between the two erasure modes by a rotation from the energy qubit $E$ to the charge  qubit $Q$, represented in (b) as two orthogonal directions. What is the general tradeoff between the different charges dissipated? The protocol provided achieves the bound of Result~\ref{lem:landauer}, i.e. $\beta \Delta H + \alpha \Delta Q = \log 2$.}
\end{figure}

The total cost of erasure splits into an energy contribution and a charge contribution; these can be expressed in terms of the single parameter $\epsilon$, the energy at which we decide to swap basis:
\be
\label{Etot}
\Delta H_{tot}(\epsilon) =  \frac{1}{\beta}\ln \left( \frac{2}{1+e^{-\beta \epsilon}} \right)- \frac{e^{-\beta \epsilon}}{1+e^{-\beta \epsilon}} \epsilon,
\ee
\be
\label{Qtot}
\Delta Q_{tot}(\epsilon)= \frac{1}{\alpha}\ln (1+ e^{-\beta \epsilon}) + \frac{e^{-\beta \epsilon}}{1+e^{-\beta \epsilon}} \frac{\beta}{\alpha}\epsilon.
\ee
The result is shown in Fig.~\ref{fig2}. Each value of $\epsilon$ defines a different protocol, corresponding to a point in the ``energy-cost'' versus ``charge cost'' graph. The protocols described yield a curve ($\Delta H_{tot} (\epsilon), \Delta Q_{tot} (\epsilon)$) for fixed inverse temperatures $\beta$ and $\alpha$. 

As expected, we recover the usual Landauer erasure bound, $\Delta H_{tot} = \beta^{-1} \ln 2$, in the $\epsilon \rightarrow \infty$ limit and the erasure at no energy of \cite{vaccaro2011information}, $\Delta Q_{tot} = \alpha^{-1} \ln 2$, in the $\epsilon \rightarrow 0$ limit. The curves in Fig.~\ref{fig2} interpolating between these two limits satisfy tightly the bound of Eq.~\eqref{eq:landauerbound}, as can be seen combining Eqs.~\eqref{Etot} and~\eqref{Qtot}:
\be
\label{eq:landauermultiple}
\beta \Delta H_{tot}(\epsilon) + \alpha \Delta Q_{tot}(\epsilon) = \ln 2, \quad \forall \epsilon \in [0,+\infty].
\ee

\begin{figure}[h]
	\includegraphics[width=5cm]{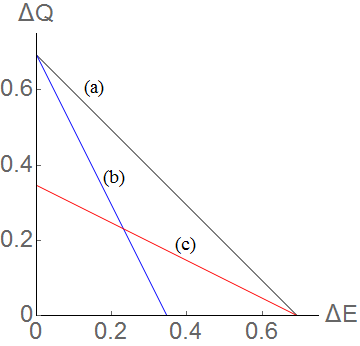}
	\caption{\label{fig2} Energy -- charge tradeoff for erasure of a memory through two baths in the protocols provided. The protocols achieve tightly the bound of Eq.~\eqref{eq:landauerbound}. Different points along each curve are parametrised by the energy $\epsilon$ at which we swap from energy erasure to charge erasure. (a), black curve: \mbox{$\beta = \alpha=1$}; (b), blue  curve: $\beta=2$, $\alpha = 1$; (c), red curve: $\beta= 1$, $\alpha=2$.}
\end{figure}

This shows how simple trade-offs exist for erasure in the presence of multiple conserved charges. One could also compare with~\cite{barnett2013beyond}, where the authors show that if erasure is performed with respect to angular momentum degrees of freedom that are not degenerate in energy, Eq.~\eqref{eq:landauermultiple} cannot be achieved.

Note that in a real experiment it might not be physically possible to modify the level structure arbitrarily, as was assumed here. This is the case for a $Q$-bath being a spin bath. One wishes to change the level structure (in angular momentum) by introducing multiple aligned spins and performing CNOT operations \cite{vaccaro2011information}. However Nature places a physical lower bound on the discrete steps, $\sim \hbar$. This does not change the general picture presented in the optimal protocol above, but introduces non-optimalities, relevant at low temperatures, which might be worth exploring (see \cite{vaccaro2011information} and Appendix~\ref{appendix_discrete}).\footnote{It is worth noticing that a major challenge for the experimental realisation of such protocol is the need for a detailed control of the interaction between the angular momentum qubit and the spin bath (see \cite{vaccaro2011information}).}

Why do such tradeoffs emerge? We now develop a broader perspective, based on the clash between different notions of equilibrium in the presence of multiple conserved quantities.

\section{Relationships between different approaches to equilibrium states}

Standard textbooks extensively discuss the thermodynamics of quantum systems with multiple, commuting conserved quantities, especially in relation to the grand canonical ensemble \cite{pathria1996statistical}.
Our focus will be on notions of equilibrium based on the maximum entropy principle and complete passivity, and we will allow for the conserved quantities to be mutually non-commuting. For related works, see \cite{halpern2016microcanonical, guryanova2016thermodynamics, perarnau2015work}.

Even though we do not focus exclusively on those, it is useful to give some simple examples of systems with mutually non-commuting conserved quantities. One is that of a rotationally symmetric Hamiltonian $H$. The angular momentum in three orthogonal directions $L_i$, \mbox{$i=1,2,3$}, are conserved quantities, $[H,L_i] =0$; of course, if $i \neq j$, $[L_i,L_j] \neq 0$. An even simpler example is the three-qubit system with Hamiltonian $H = Z \otimes Z \otimes Z$ and the mutually non-commuting conserved quantities $A = X \otimes X \otimes \I$ and $B = \I \otimes Y \otimes Y$. Assuming that the system attains some equilibrium values for energy and angular momentum $\langle H \rangle$, $\langle L_i \rangle$, $i=1,2,3$ (or $\langle H \rangle$, $\langle A \rangle$, $\langle B \rangle$ in the second example) what is the equilibrium state? 

An alternative point of view is given by the thermalisation of an isolated many-body quantum system. Strictly speaking, such systems cannot thermalise, as they evolve unitarily. However, in many scenarios, for long times sufficiently local observables will be effectively described by thermal density matrices \cite{polkovnikov2011colloquium}. In such systems the existence of local observables whose average value is conserved imposes further constraints. We ask here what happens when many of them are present, and possibly some are non-commuting.

\subsection{The maximum entropy principle}

A typical way to proceed is to invoke the \emph{maximum entropy principle} \cite{jaynes1}. Let $\mathcal{C}:=\{C_i\}$ denote the set of conserved quantities and $\bar{c}_i$ is their average value at equilibrium (by convention, $C_0:=H$, the Hamiltonian of the system being trivially a conserved change). The principle says that the equilibrium state is the solution to the problem \cite{vonneumann1927thermodynamik}\,\footnote{We give a straightforward generalisation of von Neumann's formulation that allows for multiple conserved quantities.}
\begin{eqnarray}
\label{eq:maxentformulation}
\mbox{Maximize } & S(\rho) = - \tr{}{\rho \log \rho}, \nonumber \\ \mbox{ subject to: } & \tr{}{\rho \, C_i} = \bar{c}_i, \; k=1,...,n.
\label{eq:maxentropy}
\end{eqnarray}

The solution is a so-called Generalised Gibbs Ensemble with respect to the conserved quantities, i.e.
\be
\label{eq:generalisedgibbs}
\rho_{\mathcal{C}} = e^{-\sum_i \mu_i C_i}/ \tr{}{ e^{-\sum_i \mu_i C_i}},
\ee
where $\mu_i$ are generalised chemical potential that can be easily computed as functions of the given values $\bar{c}_i$ (by the convention chosen, $\mu_0 := \beta = 1/(kT)$).\footnote{Note that the maximum entropy description is robust, in the sense that it can be applied to situations in which the system only has approximate constants of motion. For details, see \cite{kollar2011generalized}.} If $\mathcal{C}$ consists of only the Hamiltonian, one recovers the standard Gibbs ensemble. $\mathcal{C} = \{H,N\}$, where $N$ is the particle number operator, gives the grand canonical ensemble. Of course, as $N$ is by assumption a conserved quantity, \mbox{$[H,N] =0$}; hence, non-commutativity cannot appear in this problem. Non-commutativity, on the other hand, appears in the examples given at the beginning of this section.

A brief comment is necessary regarding the use of the maximum entropy principle in the presence of non-commuting quantities. In the commuting case, this can be justified as an essentially unique inference method satisfying consistency axioms in the handling of information \cite{shore}. In general, the principle can be given a microcanonical derivation. This is based on 1. Identifying expectation values with ensemble averages 2. Assigning equal probability to every element of an ensemble characterised by a total $C_i$ peaked around $\bar{c}_i$. Technical challenges arise in the non-commutativity case from the fact that the total $C_i$s are only approximately mutually commuting for any finite ensemble, so they cannot all have strictly well-defined values. These issues were tackled in \cite{balian1987equiprobability}, and recently in \cite{halpern2016microcanonical}, through the notion of an approximate microcanonical ensemble (Def.~2 Section~3 of \cite{halpern2016microcanonical} and Eq.~(2.24) of \cite{balian1987equiprobability}). This leads to the recovery of Eq.~\eqref{eq:generalisedgibbs}.

Once the maximum entropy principle is accepted, the important observation is the following: \emph{the construction works equally well for commuting or non-commuting observables}, In every case, it gives a state of the form of Eq.~\eqref{eq:generalisedgibbs}. However, there are subtler differences. Consider the projection  $\rho \mapsto \rho_{\mathcal{C}}$, 	where $\rho_{\mathcal{C}}$ is the maximum entropy state among all states $\sigma$ satisfying $\tr{}{\sigma C_i} = \tr{}{\rho C_i}$ for all $C_i \in \mathcal{C}$. This map is not necessarily completely positive when $\mathcal{C}$ contains mutually non-commuting observables (see Appendix~\ref{appendix_ME}). This agrees with the fact that the maximum entropy projection entails an inherently inferential procedure and, strictly speaking, cannot be captured purely dynamically without invoking some notion of coarse-graining or state-dependent processes. Another difference regards discontinuities induced by non-commutativity, as we mention later.

We now compare this construction to another standard thermodynamic approach based on the notions of passivity and complete passivity. We show that a disagreement emerges in the presence of multiple conserved quantities and discuss its significance.

\subsection{Complete passivity}

Among various ways of describing the content of the second law of thermodynamics, one is that no work can be extracted from a system in thermal equilibrium by means of an adiabatic process \cite{lenard1978thermodynamical}. Adiabatic here means a process where some external controls are varied for some finite time, inducing a unitary evolution generated by a time-dependent Hamiltonian $H(t)$. The total work done on the system is then traditionally defined as
\be
W(t) = \int_{0}^t \tr{}{\rho(t) \frac{dH(t)}{dt}}dt,
\ee 
where $\rho(t)$ evolves under the Schr\"odinger equation, $d\rho(t)/dt = -i [\rho(t), H(t)]$. Equivalently, one can easily check that if $U$ is the unitary evolution generated by $H(t)$ up to the final interaction time $t_F$,
\be
\label{eq:work}
W(t_F) = \tr{}{U \rho U^\dag H} - \tr{}{ \rho H }:=W_H(U,\rho),
\ee
where $\rho$ and $H$ denote, respectively, the state and the Hamiltonian at the initial time \cite{lenard1978thermodynamical}. The second law in the form stated above imposes that $\rho$, when in equilibrium, should take a form such that $W_H(U,\rho) \geq 0$ for all unitaries $U$. This can be interpreted as the fact that one can never displace a weight system upwards using the state $\rho$. If this holds  for every unitary evolution, $\rho$ is called a \emph{passive} state. A relatively straightforward computation shows that a state $\rho$ is passive if and only if $[\rho,H]=0$ and its eigenvalues are a non-increasing function of energy (i.e. there are no ``population-inversions") \cite{lenard1978thermodynamical, pusz1978passive, skrzypczyk2015passivity}.

More importantly, in 1978 an answer was formulated for the following question: what states are \emph{completely passive}, in the sense that arbitrarily many copies cannot raise a weight? More precisely, for what states $\rho$ do we have \mbox{$W_H(U,\rho^{\otimes n}) \geq 0$} for all $n \in \mathbb{N}$ and all unitaries $U$? The answer was given first from an algebraic perspective \cite{pusz1978passive} and later by Lenard using finite-dimensional methods \cite{lenard1978thermodynamical}. One consequence of these seminal works is that, for finite-dimensional systems, the Gibbs state is the only functional form (modulo limiting cases of ground states and ``infinite temperature'' systems) that does not trivialise every work process when it is assumed to be freely available.\footnote{For infinite dimensional systems such as quantum fields this notion is generalized to the set of KMS states and has a more complex structure, which we entirely ignore here.}

\subsection{Disconnection between maximum entropy construction and complete passivity}

Extending Eq.~\eqref{eq:work} to other conserved quantities beyond energy, passivity and complete passivity can be defined in an analogous way:
\begin{defn}
Given an observable $C$, a state $\rho$ is called $C$-passive if \mbox{$W_C(U,\rho) \geq 0$} for every unitary $U$. $\rho$ is $C$-completely passive if $\rho^{\otimes n}$ is $C$-passive for every $n \in \mathbb{N}$.
\end{defn}

In the case of a single conserved quantity, there is agreement between the maximum entropy and the complete passivity point of views of equilibrium. For example, a $Z$-spin bath with degenerate Hamiltonian has a maximum entropy state $\propto e^{-\alpha L_z}$, which is completely passive with respect to $L_z$, so that angular momentum in the $Z$-direction cannot be extracted from any number of copies of it by any unitary interaction.

On the other hand, let us analyse the multiple observables case. We then have
\begin{lem}[Disconnection between approaches] Let \mbox{$\mathcal{C} = \{H, C_1,...,C_n\}$}. Then, for $n \geq 1$, the maximum entropy state $\propto \exp \left[\sum_{i=0}^n \mu_i C_i\right]$ is $C(\v{\mu})$-completely passive, where $C(\v{\mu}):=\sum_{i=0}^n \mu_i C_i$, but in general it is not $C_i$-completely passive for each $i$.
	\label{lem:disagreement}
\end{lem}
\begin{proof}
Consider first the case where $C_i$ are mutually commuting. It then suffices to focus on the case where there is a single extra conserved charge $C_1$ beyond energy (with $C_1 \neq H$). In this case the maximum entropy state $\rho_{\mathcal{C}}$ reads, due to $[H,C_1]=0$,
\be
\nonumber
\rho_{\mathcal{C}} \propto e^{-\beta H} e^{-\mu_1 C_1} := \gamma_0 \gamma_1.
\ee
 A necessary condition for $\gamma_0 \gamma_1$ to be $H$-completely passive is to be $H$-passive. Passivity is equivalent to 1. [$\gamma_0 \gamma_1, H] = 0$ (which is satisfied) and 2. The eigenvalues of $\gamma_0 \gamma_1$ are monotonically decreasing for strictly increasing eigenvalues of $H$. However fix $\epsilon$ and $\tilde{\epsilon}$ with $\epsilon < \tilde{\epsilon}$. One in general can find $\ell > \tilde{\ell}$ such that $e^{-\beta \epsilon - \mu_1 \ell} > e^{-\beta \tilde{\epsilon} - \mu_1 \tilde{\ell}}$. Hence $H$-passivity is generically violated.

Consider now the non-commuting case. Given a conserved quantity $C$, a necessary condition for passivity is $[\rho, C] =0$. A useful equivalent way of thinking about this condition is that $\rho$ must be \emph{symmetric} under the $U(1)$ group generated by $C$, i.e. $e^{-i C t}\rho e^{i C t} = \rho$ for all $t$. Consider then the simple example \mbox{$\mathcal{C} = \{H, L_x, L_y, L_z\}$}, with $H$ spherically symmetric. The maximum entropy state of Eq.~\eqref{eq:generalisedgibbs} reads
\be
\label{eq:sphericalme}
\rho_{\mathcal{C}} \propto e^{-\beta H} e^{- \mu_x L_x - \mu_y L_y - \mu_z L_z}:= e^{-\beta H} e^{-\v{\mu} \cdot \v{L}},
\ee
where $\mbox{$\v{\mu}:=(\mu_x,\mu_y, \mu_z)$}$ and $\v{L} = (L_x, L_y, L_z)$. Assuming $\v{\mu} \neq \v{0}$, one can quickly check that this generalised Gibbs ensemble \emph{is not} spherically symmetric, i.e. it is not invariant under the $SU(2)$ symmetry group generated by $L_x$, $L_y$, $L_z$. In fact, $\rho_{\mathcal{C}}$ is only invariant under a $U(1)$ subgroup representing rotations about the direction $\v{\mu}$. This implies that in general the state is not passive with respect any of the conserved quantities $L_i$, and so it is not $L_i$-completely passive. 

Finally, showing that the maximum entropy state is $C(\v{\mu})$-completely passive is based on a standard argument, that we give in Appendix \ref{appendix_cmucompletepassivity}.
\end{proof}

Notice that we cannot strengthen Result~\ref{lem:disagreement} and say that $\rho_{\mathcal{C}}$ is not $C_i$-completely passive for any $i$, even under the additional assumption $\mu_i \neq 0$ for all $i$. Here is a counterexample: let $C_1 = \ketbra{+}{+}$, $C_2 = \ketbra{0}{0}$, $C_3 = \ketbra{1}{1}$ be observables on an effective qubit degree of freedom, with $\mu_1=\mu_2=\mu_3:=\mu \neq 0$ ($H$ commutes so we can leave it out of this discussion). One can check that $[C_1, C_2] \neq 0$, $[C_1,C_3] \neq 0$. But $\rho_{C} \propto e^{-\mu S}$, $S = C_1 + C_2 + C_3 = \I + C_1$, so $\rho_{\mathcal{C}}$ is completely passive with respect to $C_1$.

What Result~\ref{lem:disagreement} implies for the example of angular momentum presented above is that from arbitrarily many copies of the maximum entropy state of a system with rotationally invariant Hamiltonian we can generically extract an unbounded amount of angular momentum in every direction excluding $\v{\mu}$. We can now understand the tradeoffs of Result~\ref{lem:landauer} from Result~\ref{lem:disagreement}, noticing that the maximum entropy state used in the derivation is \mbox{$\beta \Delta H + \sum_{i\geq 1} \mu_i C_i$}-completely passive, but not $C_i$-completely passive for each $i$.

\subsection{Physical considerations on a weaker form of complete passivity}

Let us come back to the above example of a spherically symmetric Hamiltonian $H$, with $\mathcal{C} = \{H,L_x,L_y,L_z\}$, and $\rho_{\mathcal{C}}$ given by Eq.~\eqref{eq:sphericalme}. Now notice that \mbox{$\rho_{\mathcal{C}} = \rho_{\tilde{\mathcal{C}}}$}, where $\tilde{\mathcal{C}} = \{\tilde{H}(\v{\mu})\}$ and $\tilde{H}(\v{\mu}) = H + \beta^{-1} \v{\mu} \cdot \v{L}$. In fact, this is true in general, taking $\tilde{\mathcal{C}} = \{\tilde{H}(\v{\mu})\}$, \mbox{$\tilde{H}(\v{\mu}) = H + \beta^{-1} C(\v{\mu})$}.

What this means is that the maximum entropy construction gives the same state for a system with multiple observables that it would assign to a system with a single Hamiltonian $\tilde{H}(\v{\mu})$, where the Lagrange multipliers of the extremum problem play the role of coupling strengths $\mu_i$ to the various charges. In other words, we can think of $C(\v{\mu})$-complete passivity in Result~\ref{lem:disagreement} as a \emph{constrained} form of passivity, that coincides with the standard one applied only to a particular ``direction'' $\v{\mu}$. 

In the example $\mathcal{C} = \{H,L_x,L_y,L_z\}$, the generalised Gibbs state coincides with the equilibrium state of a system with the same Hamiltonian and subject to a magnetic field in the direction of $\v{\mu}$. In this dual description the Lagrange multipliers are naturally interpreted as constraint parameters defined by the physics; the residual $U(1)$ symmetry of $\rho_{\mathcal{C}}$ has an obvious interpretation as a special direction singled out by the physical problem.

There are good reasons why it may be sensible to frame the problem in this way. It is well-established that the low temperature states of interacting spin systems (such as Ising models with transverse magnetic fields) display thermodynamic phase transitions depending on the particular external field parameters, and are intrinsically quantum-mechanical in origin. We may argue that a thermodynamic approach to non-commuting conserved charges benefits from starting with this weakened form of passivity, computing its properties, and then ascertaining if subsequent variation of the parameters displays discontinuities. In fact, it is known that the maximum entropy inference can have discontinuities \cite{weis2012entropy} if some of the conserved quantities are mutually non-commuting (whereas it is continuous in the commuting case \cite{barndorff2014information}), and that these are related to quantum phase transitions \cite{chen2015discontinuity}. 

These considerations seem to favour an approach based on constrained passivity and in agreement with the maximum entropy principle. On the other hand, as we now clarify, Result~\ref{lem:disagreement} also shows that the maximum entropy state effectively acts as a (quantum) \emph{reference frame} \cite{bartlett2007reference}. We discuss why this challenges a proper inclusion of non-commutativity within the so-called resource theory approach to thermodynamics.

\section{Limitations of the resource theory approach}

\subsection{How should we build a resource theory of thermodynamics in the non-commuting case?} 

Recently a resource theory formulation of thermodynamics has been put forward, initiated in \cite{janzing2000thermodynamic, brandao2011resource}. 
Every resource theory is based on two notions: a subset of all quantum maps defines the set of \emph{allowed operations}, and a subset of all preparations defines the set of \emph{free states} $\mathcal{F}$ (these are generically assumed to be available in any number). For example, in the theory of entanglement, the free operations are Local Operations and Classical Communication (LOCC) and $\mathcal{F}$ is given by the set of separable states. Resource states are all entangled states. 

In the resource theory of Thermal Operations the allowed transformations are defined through a conservation law: they are all unitaries preserving energy; and $\mathcal{F}$ is given by all thermal states, i.e. $\gamma_R= e^{-\beta H_R}/\tr{}{e^{-\beta H_R}}$, for arbitrary $H_R$ and fixed $\beta$. We are also allowed to trace away some degrees of freedom. Combining the above, a general Thermal Operation can be written as
\be
\mathcal{E}(\rho) = \tr{R'}{U(\rho \otimes \gamma_R)U^\dag},
\ee 
with $U$ satisfying $[U, H + H_R] = 0$, $H$ the Hamiltonian of the system and $\gamma_R \in \mathcal{F}$. In general $R \neq R'$.  A detailed discussion of these assumptions and their connection to other approaches is given in Appendix~\ref{appendix_assumption}. 

Here we extend the conservation law on which Thermal Operations are defined to multiple, and in general \emph{non-commuting}, conserved quantities. For simplicity, we can limit ourselves to two conserved quantities, the extension to more charges being obvious. We define $(H,C_1,C_2)$-Thermal Operations as the set of transformations whose Stinespring dilation reads
\be
\mathcal{E}(\rho) = \tr{R'}{V(\rho \otimes \gamma_R)V^\dag},
\ee
with $V$ obeying 
\be
\nonumber
[V, X\otimes \I_R + \I\otimes X_R]= 0, \quad X=H, C_1,C_2
\ee
and $\gamma_R \in \mathcal{F}_{M}$. Here $\mathcal{F}_{M}$ is an appropriate generalisation of the set of free states to multiple and generally non-commuting conserved quantities. How do we choose $\mathcal{F}_{M}$?

In the case of a single conservation law the choice is based on the notion of non-trivial work transformations. A resource theory is called \emph{non-trivial} when some transitions $\rho \rightarrow \sigma$ are not allowed by means of free operations. An infinite set of possible non-trivial resource theories can be built from different choices of $\mathcal{F}$, given the conservation law on energy. These include the resource theory of Thermal Operations, $U(1)$-asymmetry theory (where $\mathcal{F}$ is given by all states $\rho$ invariant with respect to time translations \cite{gour2008resource}, $e^{-iH t } \rho e^{-iH t } = \rho $), and theories where $\mathcal{F}$ includes states with some specific modes of $U(1)$-asymmetry (as defined in \cite{marvian2013modes}).

A natural question is what singles out $\mathcal{F}$ in the resource theory of thermodynamics among all non-trivial theories. An answer was given in \cite{brandao2013second}, where it was shown that a state $\propto e^{-\beta H}$ is the only form ensuring non-trivial work processes. Specifically, this is the only choice of free state that does not allow to increase arbitrarily the average energy of a battery system (Theorem 8 of \cite{brandao2013second}). This is linked to the fact that the Gibbs state is, modulo limiting cases, the only completely passive state, and of course it also agrees with a choice of $\mathcal{F}$ based on the maximum entropy principle. All guiding principles suggest one and the same choice of $\mathcal{F}$. 

We are interested here in the generalisation of the resource theory of Thermal Operations in the presence of multiple conserved quantities. Due to the incompatibility between different approaches, it seems there are at least two distinct possibilities to choose $\mathcal{F}_{M}$:
\begin{enumerate}
	\item $\mathcal{F}_{M}$ should be given by states that are completely passive with respect to all thermodynamic variables.
	\item $\mathcal{F}_{M}$ should be given by the maximum entropy construction, or equivalently it should respect a constrained form of complete passivity, as established by Result~\ref{lem:disagreement}. This is the choice made in \cite{halpern2016microcanonical}.
\end{enumerate}
As we will now see, both these approaches have their limitations, highlighting current issues in the resource theory approach to thermodynamics.

\subsection{Limitations of current approaches}

If we follow the complete passivity approach, we find that the free states $\gamma$ should satisfy $[\gamma, C_i] = 0$ for all $C_i$. In the generic case this implies
\be
\gamma =  \sum_k a_k O_k, \mbox{  with } a_k \in \mathbb{R},
\ee
for some collection of observables \mbox{$O_k \in \cap_i \rm{Com}(C_i)$}, where \rm{Com}$(C_i)$ is the commutant of the operator $C_i$. The significance of this is that a $(C_1,...,C_n)$-passive state \emph{cannot} contain any component of $C_i$ in it. To explain what this means, for simplicity consider the case of a trivial Hamiltonian, and two non-commuting operators $A$ and $B$ such that the only operator that commutes with both is one proportional to the identity (as in the qubit example with the Pauli observables $A=X$ and $B=Y$). In this case the only passive or completely passive state with respect to both observables is the maximally mixed state $\I/d$. In particular, it is \emph{impossible} to reproduce the maximum entropy Gibbsian distribution in the $A$ and $B$ degrees of freedom. This is not to say that a resource theory is impossible, but the bath states act as random noise in the non-commuting charge degrees of freedom, in disagreement with the maximum entropy principle and equilibration based on typicality arguments. 

Notice that the choice of free states based on complete passivity is in agreement with asymmetry theory, as $\gamma$ is required to be symmetric with respect to the smallest group $G$ generated by the non-commuting observables at hand. More specifically, $U_g \rho U_g^{\dag} = \rho$ for all $g \in G$, where $g \mapsto U_g$ is an appropriate unitary representation of $G$. Any state not invariant under the action of $G$ is called a \emph{reference frame}, in that it can be used to encode information about the group element $g \in G$ (for $G=SU(2)$, $g$ is a direction) \cite{bartlett2007reference}. For the qubit example with $A=X$ and $B=Y$, we have $G=SU(2)$, so the only symmetric state is the maximally mixed one. Any $\rho \neq \I/2$ breaks rotational symmetry. From a symmetry perspective, any non symmetric state is a resource, and complete passivity ensures that no reference frame for $G$ is being introduced. 

Having a thermodynamic application in mind, and given the strong constraints imposed by complete passivity, one may be tempted to follow the second route and select $\mathcal{F}_{M}$ on the basis of the maximum entropy principle or, equivalently, a constrained notion of complete passivity with respect to $C(\v{\mu})$ (an approach followed in \cite{halpern2016microcanonical}). As the previous discussion implies, however, this choice of $\mathcal{F}_{M}$ allows the free introduction of reference frames for the group $G$ generated by the non-commuting charges. The point is easily understood in the example of Eq. \eqref{eq:sphericalme}, noticing that the maximum entropy states are not rotationally symmetric. The states in $\mathcal{F}_{M}$ are only invariant under a $U(1)$ subgroup of $G$. The symmetry constraints imposed by the conservation of multiple conserved quantities are then partially lifted, with exclusion of an abelian (specifically, $U(1)$) subgroup. Thermodynamically this implies that only the average value of $\beta  H + \sum_{i=1}^n \mu_i C_i$ cannot be increased indefinitely acting on a large number of free states.

 Hence to recover a clear thermodynamic interpretation we do not fully capture the underlying non-commutativity. To summarise, one has to choose between
\begin{enumerate}
	\item No external, symmetry-breaking axis, but only a limited Gibbsian form.
	\item A symmetry-broken scenario with a clear thermodynamic interpretation. 
\end{enumerate}

Despite some advances \cite{halpern2016microcanonical, guryanova2016thermodynamics} we hence believe that it is still an open question how resource theory approaches can capture the core aspects of the role of non-commutativity in thermodynamic processes. One reason at the root of some of these difficulties may be that resource theories do not easily handle external fields \cite{frenzel2014reexamination}.

\section{Conclusions}

We have presented an overview of some subtleties of thermodynamics in the presence of multiple conserved quantities. We proved that for some results, such as Landauer erasure, this entails a formally simple modification to take into account the presence of multiple ``currencies'' that can be used to pay for erasure. 

We have shown that these tradeoffs can also be seen through the lens of a more general result: the disagreement between maximum entropy construction and complete passivity, that in the single-charge scenario  give equivalent characterisations of the notion of equilibrium. 

This also lead us to reconsider the role of non-commutativity in current resource theory approaches. The seemingly most reasonable choice for the set of free states (the one based on the maximum entropy construction) is only equivalent to a constrained notion of complete passivity where a specific directions is singled out as special -- effectively making the theory insensitive to the stronger form of non-commutativity. The maximum entropy state acts as a quantum reference frame (in the sense of \cite{bartlett2007reference}) for the non-commuting observables, with exclusion of an abelian subgroup.

Non-commutativity, already well-established in the context of the study of quantum phase transitions, is an exciting new venue for small-scale thermodynamics.
Understanding to what extent it can be captured in the resource theory framework is a relevant open question.

\bigskip

\begin{acknowledgments}We (warmly) thank Kamil Korzekwa for numerous (heated) discussions on these topics. We would also like to thank Raam Uzdin for insightful discussions and the concise entropic account of passivity in \cite{uzdin2015coherence}. ML is supported by EPSRC and in part by COST Action MP1209. ML acknowledges financial support from the Spanish Ministry of Economy and Competitiveness, through the “Severo Ochoa” Programme for Centres of Excellence in R\&D (SEV-2015-0522), Spanish MINECO (QIBEQI FIS2016-80773-P), Fundació Cellex, CERCA Programme / Generalitat de Catalunya. DJ is supported by the Royal Society. 
	
\end{acknowledgments}

\bibliography{Bibliography_thermodynamics_4}

\begin{thebibliography}{36}%
\makeatletter
\providecommand \@ifxundefined [1]{%
 \@ifx{#1\undefined}
}%
\providecommand \@ifnum [1]{%
 \ifnum #1\expandafter \@firstoftwo
 \else \expandafter \@secondoftwo
 \fi
}%
\providecommand \@ifx [1]{%
 \ifx #1\expandafter \@firstoftwo
 \else \expandafter \@secondoftwo
 \fi
}%
\providecommand \natexlab [1]{#1}%
\providecommand \enquote  [1]{``#1''}%
\providecommand \bibnamefont  [1]{#1}%
\providecommand \bibfnamefont [1]{#1}%
\providecommand \citenamefont [1]{#1}%
\providecommand \href@noop [0]{\@secondoftwo}%
\providecommand \href [0]{\begingroup \@sanitize@url \@href}%
\providecommand \@href[1]{\@@startlink{#1}\@@href}%
\providecommand \@@href[1]{\endgroup#1\@@endlink}%
\providecommand \@sanitize@url [0]{\catcode `\\12\catcode `\$12\catcode
  `\&12\catcode `\#12\catcode `\^12\catcode `\_12\catcode `\%12\relax}%
\providecommand \@@startlink[1]{}%
\providecommand \@@endlink[0]{}%
\providecommand \url  [0]{\begingroup\@sanitize@url \@url }%
\providecommand \@url [1]{\endgroup\@href {#1}{\urlprefix }}%
\providecommand \urlprefix  [0]{URL }%
\providecommand \Eprint [0]{\href }%
\providecommand \doibase [0]{http://dx.doi.org/}%
\providecommand \selectlanguage [0]{\@gobble}%
\providecommand \bibinfo  [0]{\@secondoftwo}%
\providecommand \bibfield  [0]{\@secondoftwo}%
\providecommand \translation [1]{[#1]}%
\providecommand \BibitemOpen [0]{}%
\providecommand \bibitemStop [0]{}%
\providecommand \bibitemNoStop [0]{.\EOS\space}%
\providecommand \EOS [0]{\spacefactor3000\relax}%
\providecommand \BibitemShut  [1]{\csname bibitem#1\endcsname}%
\let\auto@bib@innerbib\@empty
\bibitem [{\citenamefont {Callen}(1985)}]{callen1985thermodynamics}%
  \BibitemOpen
  \bibfield  {author} {\bibinfo {author} {\bibfnamefont {H.}~\bibnamefont
  {Callen}},\ }\href@noop {} {\emph {\bibinfo {title} {Thermodynamics and an
  Introduction to Thermostatistics -- Second Edition}}}\ (\bibinfo  {publisher}
  {John Wiley \& Sons},\ \bibinfo {year} {1985})\BibitemShut {NoStop}%
\bibitem [{\citenamefont {{Plenio}}\ and\ \citenamefont
  {{Vitelli}}(2001)}]{plenio2001physics}%
  \BibitemOpen
  \bibfield  {author} {\bibinfo {author} {\bibfnamefont {M.~B.}\ \bibnamefont
  {{Plenio}}}\ and\ \bibinfo {author} {\bibfnamefont {V.}~\bibnamefont
  {{Vitelli}}},\ }\href {\doibase 10.1080/00107510010018916} {\bibfield
  {journal} {\bibinfo  {journal} {Contemporary Physics}\ }\textbf {\bibinfo
  {volume} {42}},\ \bibinfo {pages} {25} (\bibinfo {year} {2001})}\BibitemShut
  {NoStop}%
\bibitem [{\citenamefont {Lostaglio}(2013)}]{lostagliothesis}%
  \BibitemOpen
  \bibfield  {author} {\bibinfo {author} {\bibfnamefont {M.}~\bibnamefont
  {Lostaglio}},\ }\href@noop {} {\bibfield  {journal} {\bibinfo  {journal}
  {Generalised Landauer Principle, MRes thesis, Imperial College London}\ }
  (\bibinfo {year} {2013})}\BibitemShut {NoStop}%
\bibitem [{\citenamefont {{Reeb}}\ and\ \citenamefont
  {{Wolf}}(2014)}]{reeb2013proving}%
  \BibitemOpen
  \bibfield  {author} {\bibinfo {author} {\bibfnamefont {D.}~\bibnamefont
  {{Reeb}}}\ and\ \bibinfo {author} {\bibfnamefont {M.~M.}\ \bibnamefont
  {{Wolf}}},\ }\href {\doibase 10.1088/1367-2630/16/10/103011} {\bibfield
  {journal} {\bibinfo  {journal} {New Journal of Physics}\ }\textbf {\bibinfo
  {volume} {16}},\ \bibinfo {eid} {103011} (\bibinfo {year}
  {2014})}\BibitemShut {NoStop}%
\bibitem [{\citenamefont {Oppenheim}\ \emph {et~al.}(2002)\citenamefont
  {Oppenheim}, \citenamefont {Horodecki}, \citenamefont {Horodecki},\ and\
  \citenamefont {Horodecki}}]{oppenheim2002approach}%
  \BibitemOpen
  \bibfield  {author} {\bibinfo {author} {\bibfnamefont {J.}~\bibnamefont
  {Oppenheim}}, \bibinfo {author} {\bibfnamefont {M.}~\bibnamefont
  {Horodecki}}, \bibinfo {author} {\bibfnamefont {P.}~\bibnamefont
  {Horodecki}}, \ and\ \bibinfo {author} {\bibfnamefont {R.}~\bibnamefont
  {Horodecki}},\ }\href {\doibase 10.1103/PhysRevLett.89.180402} {\bibfield
  {journal} {\bibinfo  {journal} {Phys. Rev. Lett.}\ }\textbf {\bibinfo
  {volume} {89}},\ \bibinfo {pages} {180402} (\bibinfo {year}
  {2002})}\BibitemShut {NoStop}%
\bibitem [{\citenamefont {{del Rio}}\ \emph {et~al.}(2011)\citenamefont {{del
  Rio}}, \citenamefont {{Aberg}}, \citenamefont {{Renner}}, \citenamefont
  {{Dahlsten}},\ and\ \citenamefont {{Vedral}}}]{rio2010thermodynamic}%
  \BibitemOpen
  \bibfield  {author} {\bibinfo {author} {\bibfnamefont {L.}~\bibnamefont {{del
  Rio}}}, \bibinfo {author} {\bibfnamefont {J.}~\bibnamefont {{Aberg}}},
  \bibinfo {author} {\bibfnamefont {R.}~\bibnamefont {{Renner}}}, \bibinfo
  {author} {\bibfnamefont {O.}~\bibnamefont {{Dahlsten}}}, \ and\ \bibinfo
  {author} {\bibfnamefont {V.}~\bibnamefont {{Vedral}}},\ }\href {\doibase
  10.1038/nature10123} {\bibfield  {journal} {\bibinfo  {journal} {Nature}\
  }\textbf {\bibinfo {volume} {474}} (\bibinfo {year} {2011}),\
  10.1038/nature10123}\BibitemShut {NoStop}%
\bibitem [{\citenamefont {Perarnau-Llobet}\ \emph {et~al.}(2015)\citenamefont
  {Perarnau-Llobet}, \citenamefont {Hovhannisyan}, \citenamefont {Huber},
  \citenamefont {Skrzypczyk}, \citenamefont {Brunner},\ and\ \citenamefont
  {Ac\'{\i}n}}]{llobet2015extractable}%
  \BibitemOpen
  \bibfield  {author} {\bibinfo {author} {\bibfnamefont {M.}~\bibnamefont
  {Perarnau-Llobet}}, \bibinfo {author} {\bibfnamefont {K.~V.}\ \bibnamefont
  {Hovhannisyan}}, \bibinfo {author} {\bibfnamefont {M.}~\bibnamefont {Huber}},
  \bibinfo {author} {\bibfnamefont {P.}~\bibnamefont {Skrzypczyk}}, \bibinfo
  {author} {\bibfnamefont {N.}~\bibnamefont {Brunner}}, \ and\ \bibinfo
  {author} {\bibfnamefont {A.}~\bibnamefont {Ac\'{\i}n}},\ }\href {\doibase
  10.1103/PhysRevX.5.041011} {\bibfield  {journal} {\bibinfo  {journal} {Phys.
  Rev. X}\ }\textbf {\bibinfo {volume} {5}},\ \bibinfo {pages} {041011}
  (\bibinfo {year} {2015})}\BibitemShut {NoStop}%
\bibitem [{\citenamefont {{Vaccaro}}\ and\ \citenamefont
  {{Barnett}}(2011)}]{vaccaro2011information}%
  \BibitemOpen
  \bibfield  {author} {\bibinfo {author} {\bibfnamefont {J.~A.}\ \bibnamefont
  {{Vaccaro}}}\ and\ \bibinfo {author} {\bibfnamefont {S.~M.}\ \bibnamefont
  {{Barnett}}},\ }\href {\doibase 10.1098/rspa.2010.0577} {\bibfield  {journal}
  {\bibinfo  {journal} {Proceedings of the Royal Society of London Series A}\
  }\textbf {\bibinfo {volume} {467}},\ \bibinfo {pages} {1770} (\bibinfo {year}
  {2011})}\BibitemShut {NoStop}%
\bibitem [{\citenamefont {{Barnett}}\ and\ \citenamefont
  {{Vaccaro}}(2013)}]{barnett2013beyond}%
  \BibitemOpen
  \bibfield  {author} {\bibinfo {author} {\bibfnamefont {S.}~\bibnamefont
  {{Barnett}}}\ and\ \bibinfo {author} {\bibfnamefont {J.}~\bibnamefont
  {{Vaccaro}}},\ }\href {\doibase 10.3390/e15114956} {\bibfield  {journal}
  {\bibinfo  {journal} {Entropy}\ }\textbf {\bibinfo {volume} {15}},\ \bibinfo
  {pages} {4956} (\bibinfo {year} {2013})}\BibitemShut {NoStop}%
\bibitem [{\citenamefont {Pathria}(1996)}]{pathria1996statistical}%
  \BibitemOpen
  \bibfield  {author} {\bibinfo {author} {\bibfnamefont {R.}~\bibnamefont
  {Pathria}},\ }\href@noop {} {\emph {\bibinfo {title} {Statistical
  Mechanics}}}\ (\bibinfo  {publisher} {Elsevier},\ \bibinfo {year}
  {1996})\BibitemShut {NoStop}%
\bibitem [{\citenamefont {{Yunger Halpern}}\ \emph {et~al.}(2016)\citenamefont
  {{Yunger Halpern}}, \citenamefont {{Faist}}, \citenamefont {{Oppenheim}},\
  and\ \citenamefont {{Winter}}}]{halpern2016microcanonical}%
  \BibitemOpen
  \bibfield  {author} {\bibinfo {author} {\bibfnamefont {N.}~\bibnamefont
  {{Yunger Halpern}}}, \bibinfo {author} {\bibfnamefont {P.}~\bibnamefont
  {{Faist}}}, \bibinfo {author} {\bibfnamefont {J.}~\bibnamefont
  {{Oppenheim}}}, \ and\ \bibinfo {author} {\bibfnamefont {A.}~\bibnamefont
  {{Winter}}},\ }\href {\doibase 10.1038/ncomms12051} {\bibfield  {journal}
  {\bibinfo  {journal} {Nature Communications}\ }\textbf {\bibinfo {volume}
  {7}},\ \bibinfo {eid} {12051} (\bibinfo {year} {2016})}\BibitemShut {NoStop}%
\bibitem [{\citenamefont {{Guryanova}}\ \emph {et~al.}(2016)\citenamefont
  {{Guryanova}}, \citenamefont {{Popescu}}, \citenamefont {{Short}},
  \citenamefont {{Silva}},\ and\ \citenamefont
  {{Skrzypczyk}}}]{guryanova2016thermodynamics}%
  \BibitemOpen
  \bibfield  {author} {\bibinfo {author} {\bibfnamefont {Y.}~\bibnamefont
  {{Guryanova}}}, \bibinfo {author} {\bibfnamefont {S.}~\bibnamefont
  {{Popescu}}}, \bibinfo {author} {\bibfnamefont {A.~J.}\ \bibnamefont
  {{Short}}}, \bibinfo {author} {\bibfnamefont {R.}~\bibnamefont {{Silva}}}, \
  and\ \bibinfo {author} {\bibfnamefont {P.}~\bibnamefont {{Skrzypczyk}}},\
  }\href {\doibase 10.1038/ncomms12049} {\bibfield  {journal} {\bibinfo
  {journal} {Nature Communications}\ }\textbf {\bibinfo {volume} {7}},\
  \bibinfo {eid} {12049} (\bibinfo {year} {2016})}\BibitemShut {NoStop}%
\bibitem [{\citenamefont {Perarnau-Llobet}\ \emph {et~al.}(2016)\citenamefont
  {Perarnau-Llobet}, \citenamefont {Riera}, \citenamefont {Gallego},
  \citenamefont {Wilming},\ and\ \citenamefont {Eisert}}]{perarnau2015work}%
  \BibitemOpen
  \bibfield  {author} {\bibinfo {author} {\bibfnamefont {M.}~\bibnamefont
  {Perarnau-Llobet}}, \bibinfo {author} {\bibfnamefont {A.}~\bibnamefont
  {Riera}}, \bibinfo {author} {\bibfnamefont {R.}~\bibnamefont {Gallego}},
  \bibinfo {author} {\bibfnamefont {H.}~\bibnamefont {Wilming}}, \ and\
  \bibinfo {author} {\bibfnamefont {J.}~\bibnamefont {Eisert}},\ }\href
  {\doibase 10.1088/1367-2630/aa4fa6} {\bibfield  {journal} {\bibinfo
  {journal} {New Journal of Physics}\ }\textbf {\bibinfo {volume} {18}},\
  \bibinfo {pages} {123035} (\bibinfo {year} {2016})}\BibitemShut {NoStop}%
\bibitem [{\citenamefont {Polkovnikov}\ \emph {et~al.}(2011)\citenamefont
  {Polkovnikov}, \citenamefont {Sengupta}, \citenamefont {Silva},\ and\
  \citenamefont {Vengalattore}}]{polkovnikov2011colloquium}%
  \BibitemOpen
  \bibfield  {author} {\bibinfo {author} {\bibfnamefont {A.}~\bibnamefont
  {Polkovnikov}}, \bibinfo {author} {\bibfnamefont {K.}~\bibnamefont
  {Sengupta}}, \bibinfo {author} {\bibfnamefont {A.}~\bibnamefont {Silva}}, \
  and\ \bibinfo {author} {\bibfnamefont {M.}~\bibnamefont {Vengalattore}},\
  }\href {\doibase 10.1103/RevModPhys.83.863} {\bibfield  {journal} {\bibinfo
  {journal} {Rev. Mod. Phys.}\ }\textbf {\bibinfo {volume} {83}},\ \bibinfo
  {pages} {863} (\bibinfo {year} {2011})}\BibitemShut {NoStop}%
\bibitem [{\citenamefont {Jaynes}(1957)}]{jaynes1}%
  \BibitemOpen
  \bibfield  {author} {\bibinfo {author} {\bibfnamefont {E.~T.}\ \bibnamefont
  {Jaynes}},\ }\href {\doibase 10.1103/PhysRev.106.620} {\bibfield  {journal}
  {\bibinfo  {journal} {Phys. Rev.}\ }\textbf {\bibinfo {volume} {106}},\
  \bibinfo {pages} {620} (\bibinfo {year} {1957})}\BibitemShut {NoStop}%
\bibitem [{\citenamefont {Von~Neumann}(1927)}]{vonneumann1927thermodynamik}%
  \BibitemOpen
  \bibfield  {author} {\bibinfo {author} {\bibfnamefont {J.}~\bibnamefont
  {Von~Neumann}},\ }\href@noop {} {\bibfield  {journal} {\bibinfo  {journal}
  {Nachrichten von der Gesellschaft der Wissenschaften zu G{\"o}ttingen,
  Mathematisch-Physikalische Klasse}\ }\textbf {\bibinfo {volume} {1927}},\
  \bibinfo {pages} {273} (\bibinfo {year} {1927})}\BibitemShut {NoStop}%
\bibitem [{\citenamefont {Kollar}\ \emph {et~al.}(2011)\citenamefont {Kollar},
  \citenamefont {Wolf},\ and\ \citenamefont
  {Eckstein}}]{kollar2011generalized}%
  \BibitemOpen
  \bibfield  {author} {\bibinfo {author} {\bibfnamefont {M.}~\bibnamefont
  {Kollar}}, \bibinfo {author} {\bibfnamefont {F.~A.}\ \bibnamefont {Wolf}}, \
  and\ \bibinfo {author} {\bibfnamefont {M.}~\bibnamefont {Eckstein}},\ }\href
  {\doibase 10.1103/PhysRevB.84.054304} {\bibfield  {journal} {\bibinfo
  {journal} {Phys. Rev. B}\ }\textbf {\bibinfo {volume} {84}},\ \bibinfo
  {pages} {054304} (\bibinfo {year} {2011})}\BibitemShut {NoStop}%
\bibitem [{\citenamefont {{Shore, J. E.}}\ and\ \citenamefont {{Johnson,
  R.W.}}(1980)}]{shore}%
  \BibitemOpen
  \bibfield  {author} {\bibinfo {author} {\bibnamefont {{Shore, J. E.}}}\ and\
  \bibinfo {author} {\bibnamefont {{Johnson, R.W.}}},\ }\href@noop {}
  {\bibfield  {journal} {\bibinfo  {journal} {IEEE Transactions on Information
  Theory}\ }\textbf {\bibinfo {volume} {26}},\ \bibinfo {pages} {26} (\bibinfo
  {year} {1980})}\BibitemShut {NoStop}%
\bibitem [{\citenamefont {Balian}\ and\ \citenamefont
  {Balazs}(1987)}]{balian1987equiprobability}%
  \BibitemOpen
  \bibfield  {author} {\bibinfo {author} {\bibfnamefont {R.}~\bibnamefont
  {Balian}}\ and\ \bibinfo {author} {\bibfnamefont {N.}~\bibnamefont
  {Balazs}},\ }\href@noop {} {\bibfield  {journal} {\bibinfo  {journal} {Annals
  of Physics}\ }\textbf {\bibinfo {volume} {179}},\ \bibinfo {pages} {97}
  (\bibinfo {year} {1987})}\BibitemShut {NoStop}%
\bibitem [{\citenamefont {Lenard}(1978)}]{lenard1978thermodynamical}%
  \BibitemOpen
  \bibfield  {author} {\bibinfo {author} {\bibfnamefont {A.}~\bibnamefont
  {Lenard}},\ }\href {\doibase 10.1007/BF01011769} {\bibfield  {journal}
  {\bibinfo  {journal} {Journal of Statistical Physics}\ }\textbf {\bibinfo
  {volume} {19}},\ \bibinfo {pages} {575} (\bibinfo {year} {1978})}\BibitemShut
  {NoStop}%
\bibitem [{\citenamefont {Pusz}\ and\ \citenamefont
  {Woronowicz}(1978)}]{pusz1978passive}%
  \BibitemOpen
  \bibfield  {author} {\bibinfo {author} {\bibfnamefont {W.}~\bibnamefont
  {Pusz}}\ and\ \bibinfo {author} {\bibfnamefont {S.}~\bibnamefont
  {Woronowicz}},\ }\href {\doibase 10.1007/BF01614224} {\bibfield  {journal}
  {\bibinfo  {journal} {Communications in Mathematical Physics}\ }\textbf
  {\bibinfo {volume} {58}},\ \bibinfo {pages} {273} (\bibinfo {year}
  {1978})}\BibitemShut {NoStop}%
\bibitem [{\citenamefont {{Skrzypczyk}}\ \emph {et~al.}(2015)\citenamefont
  {{Skrzypczyk}}, \citenamefont {{Silva}},\ and\ \citenamefont
  {{Brunner}}}]{skrzypczyk2015passivity}%
  \BibitemOpen
  \bibfield  {author} {\bibinfo {author} {\bibfnamefont {P.}~\bibnamefont
  {{Skrzypczyk}}}, \bibinfo {author} {\bibfnamefont {R.}~\bibnamefont
  {{Silva}}}, \ and\ \bibinfo {author} {\bibfnamefont {N.}~\bibnamefont
  {{Brunner}}},\ }\href {\doibase 10.1103/PhysRevE.91.052133} {\bibfield
  {journal} {\bibinfo  {journal} {\pre}\ }\textbf {\bibinfo {volume} {91}},\
  \bibinfo {eid} {052133} (\bibinfo {year} {2015})}\BibitemShut {NoStop}%
\bibitem [{\citenamefont {{Weis}}\ and\ \citenamefont
  {{Knauf}}(2012)}]{weis2012entropy}%
  \BibitemOpen
  \bibfield  {author} {\bibinfo {author} {\bibfnamefont {S.}~\bibnamefont
  {{Weis}}}\ and\ \bibinfo {author} {\bibfnamefont {A.}~\bibnamefont
  {{Knauf}}},\ }\href {\doibase 10.1063/1.4757652} {\bibfield  {journal}
  {\bibinfo  {journal} {Journal of Mathematical Physics}\ }\textbf {\bibinfo
  {volume} {53}},\ \bibinfo {pages} {102206} (\bibinfo {year}
  {2012})}\BibitemShut {NoStop}%
\bibitem [{\citenamefont {Barndorff-Nielsen}(2014)}]{barndorff2014information}%
  \BibitemOpen
  \bibfield  {author} {\bibinfo {author} {\bibfnamefont {O.}~\bibnamefont
  {Barndorff-Nielsen}},\ }\href@noop {} {\emph {\bibinfo {title} {Information
  and exponential families in statistical theory}}}\ (\bibinfo  {publisher}
  {John Wiley \& Sons},\ \bibinfo {year} {2014})\BibitemShut {NoStop}%
\bibitem [{\citenamefont {{Chen}}\ \emph {et~al.}(2015)\citenamefont {{Chen}},
  \citenamefont {{Ji}}, \citenamefont {{Li}}, \citenamefont {{Poon}},
  \citenamefont {{Shen}}, \citenamefont {{Yu}}, \citenamefont {{Zeng}},\ and\
  \citenamefont {{Zhou}}}]{chen2015discontinuity}%
  \BibitemOpen
  \bibfield  {author} {\bibinfo {author} {\bibfnamefont {J.}~\bibnamefont
  {{Chen}}}, \bibinfo {author} {\bibfnamefont {Z.}~\bibnamefont {{Ji}}},
  \bibinfo {author} {\bibfnamefont {C.-K.}\ \bibnamefont {{Li}}}, \bibinfo
  {author} {\bibfnamefont {Y.-T.}\ \bibnamefont {{Poon}}}, \bibinfo {author}
  {\bibfnamefont {Y.}~\bibnamefont {{Shen}}}, \bibinfo {author} {\bibfnamefont
  {N.}~\bibnamefont {{Yu}}}, \bibinfo {author} {\bibfnamefont {B.}~\bibnamefont
  {{Zeng}}}, \ and\ \bibinfo {author} {\bibfnamefont {D.}~\bibnamefont
  {{Zhou}}},\ }\href {\doibase 10.1088/1367-2630/17/8/083019} {\bibfield
  {journal} {\bibinfo  {journal} {New Journal of Physics}\ }\textbf {\bibinfo
  {volume} {17}},\ \bibinfo {eid} {083019} (\bibinfo {year}
  {2015})}\BibitemShut {NoStop}%
\bibitem [{\citenamefont {{Bartlett}}\ \emph {et~al.}(2007)\citenamefont
  {{Bartlett}}, \citenamefont {{Rudolph}},\ and\ \citenamefont
  {{Spekkens}}}]{bartlett2007reference}%
  \BibitemOpen
  \bibfield  {author} {\bibinfo {author} {\bibfnamefont {S.~D.}\ \bibnamefont
  {{Bartlett}}}, \bibinfo {author} {\bibfnamefont {T.}~\bibnamefont
  {{Rudolph}}}, \ and\ \bibinfo {author} {\bibfnamefont {R.~W.}\ \bibnamefont
  {{Spekkens}}},\ }\href {\doibase 10.1103/RevModPhys.79.555} {\bibfield
  {journal} {\bibinfo  {journal} {Reviews of Modern Physics}\ }\textbf
  {\bibinfo {volume} {79}},\ \bibinfo {pages} {555} (\bibinfo {year}
  {2007})}\BibitemShut {NoStop}%
\bibitem [{\citenamefont {{Janzing}}\ \emph {et~al.}(2000)\citenamefont
  {{Janzing}}, \citenamefont {{Wocjan}}, \citenamefont {{Zeier}}, \citenamefont
  {{Geiss}},\ and\ \citenamefont {{Beth}}}]{janzing2000thermodynamic}%
  \BibitemOpen
  \bibfield  {author} {\bibinfo {author} {\bibfnamefont {D.}~\bibnamefont
  {{Janzing}}}, \bibinfo {author} {\bibfnamefont {P.}~\bibnamefont {{Wocjan}}},
  \bibinfo {author} {\bibfnamefont {R.}~\bibnamefont {{Zeier}}}, \bibinfo
  {author} {\bibfnamefont {R.}~\bibnamefont {{Geiss}}}, \ and\ \bibinfo
  {author} {\bibfnamefont {T.}~\bibnamefont {{Beth}}},\ }\href {\doibase
  10.1023/A:1026422630734} {\bibfield  {journal} {\bibinfo  {journal} {Int. J.
  Theor. Phys.}\ }\textbf {\bibinfo {volume} {39}},\ \bibinfo {pages} {2717}
  (\bibinfo {year} {2000})}\BibitemShut {NoStop}%
\bibitem [{\citenamefont {Brand\~ao}\ \emph {et~al.}(2013)\citenamefont
  {Brand\~ao}, \citenamefont {Horodecki}, \citenamefont {Oppenheim},
  \citenamefont {Renes},\ and\ \citenamefont {Spekkens}}]{brandao2011resource}%
  \BibitemOpen
  \bibfield  {author} {\bibinfo {author} {\bibfnamefont {F.~G. S.~L.}\
  \bibnamefont {Brand\~ao}}, \bibinfo {author} {\bibfnamefont {M.}~\bibnamefont
  {Horodecki}}, \bibinfo {author} {\bibfnamefont {J.}~\bibnamefont
  {Oppenheim}}, \bibinfo {author} {\bibfnamefont {J.~M.}\ \bibnamefont
  {Renes}}, \ and\ \bibinfo {author} {\bibfnamefont {R.~W.}\ \bibnamefont
  {Spekkens}},\ }\href {\doibase 10.1103/PhysRevLett.111.250404} {\bibfield
  {journal} {\bibinfo  {journal} {Phys. Rev. Lett.}\ }\textbf {\bibinfo
  {volume} {111}},\ \bibinfo {pages} {250404} (\bibinfo {year}
  {2013})}\BibitemShut {NoStop}%
\bibitem [{\citenamefont {{Gour}}\ and\ \citenamefont
  {{Spekkens}}(2008)}]{gour2008resource}%
  \BibitemOpen
  \bibfield  {author} {\bibinfo {author} {\bibfnamefont {G.}~\bibnamefont
  {{Gour}}}\ and\ \bibinfo {author} {\bibfnamefont {R.~W.}\ \bibnamefont
  {{Spekkens}}},\ }\href {\doibase 10.1088/1367-2630/10/3/033023} {\bibfield
  {journal} {\bibinfo  {journal} {New J. Phys.}\ }\textbf {\bibinfo {volume}
  {10}},\ \bibinfo {eid} {033023} (\bibinfo {year} {2008})}\BibitemShut
  {NoStop}%
\bibitem [{\citenamefont {{Marvian}}\ and\ \citenamefont
  {{Spekkens}}(2014)}]{marvian2013modes}%
  \BibitemOpen
  \bibfield  {author} {\bibinfo {author} {\bibfnamefont {I.}~\bibnamefont
  {{Marvian}}}\ and\ \bibinfo {author} {\bibfnamefont {R.~W.}\ \bibnamefont
  {{Spekkens}}},\ }\href {\doibase 10.1103/PhysRevA.90.062110} {\bibfield
  {journal} {\bibinfo  {journal} {\pra}\ }\textbf {\bibinfo {volume} {90}},\
  \bibinfo {eid} {062110} (\bibinfo {year} {2014})}\BibitemShut {NoStop}%
\bibitem [{\citenamefont {{Brand{\~a}o}}\ \emph {et~al.}(2015)\citenamefont
  {{Brand{\~a}o}}, \citenamefont {{Horodecki}}, \citenamefont {{Ng}},
  \citenamefont {{Oppenheim}},\ and\ \citenamefont
  {{Wehner}}}]{brandao2013second}%
  \BibitemOpen
  \bibfield  {author} {\bibinfo {author} {\bibfnamefont {F.}~\bibnamefont
  {{Brand{\~a}o}}}, \bibinfo {author} {\bibfnamefont {M.}~\bibnamefont
  {{Horodecki}}}, \bibinfo {author} {\bibfnamefont {N.}~\bibnamefont {{Ng}}},
  \bibinfo {author} {\bibfnamefont {J.}~\bibnamefont {{Oppenheim}}}, \ and\
  \bibinfo {author} {\bibfnamefont {S.}~\bibnamefont {{Wehner}}},\ }\href
  {\doibase 10.1073/pnas.1411728112} {\bibfield  {journal} {\bibinfo  {journal}
  {Proceedings of the National Academy of Science}\ }\textbf {\bibinfo {volume}
  {112}},\ \bibinfo {pages} {3275} (\bibinfo {year} {2015})}\BibitemShut
  {NoStop}%
\bibitem [{\citenamefont {Frenzel}\ \emph {et~al.}(2014)\citenamefont
  {Frenzel}, \citenamefont {Jennings},\ and\ \citenamefont
  {Rudolph}}]{frenzel2014reexamination}%
  \BibitemOpen
  \bibfield  {author} {\bibinfo {author} {\bibfnamefont {M.~F.}\ \bibnamefont
  {Frenzel}}, \bibinfo {author} {\bibfnamefont {D.}~\bibnamefont {Jennings}}, \
  and\ \bibinfo {author} {\bibfnamefont {T.}~\bibnamefont {Rudolph}},\ }\href
  {\doibase 10.1103/PhysRevE.90.052136} {\bibfield  {journal} {\bibinfo
  {journal} {Phys. Rev. E}\ }\textbf {\bibinfo {volume} {90}},\ \bibinfo
  {pages} {052136} (\bibinfo {year} {2014})}\BibitemShut {NoStop}%
\bibitem [{\citenamefont {{Uzdin}}(2016)}]{uzdin2015coherence}%
  \BibitemOpen
  \bibfield  {author} {\bibinfo {author} {\bibfnamefont {R.}~\bibnamefont
  {{Uzdin}}},\ }\href {\doibase 10.1103/PhysRevApplied.6.024004} {\bibfield
  {journal} {\bibinfo  {journal} {Physical Review Applied}\ }\textbf {\bibinfo
  {volume} {6}},\ \bibinfo {eid} {024004} (\bibinfo {year} {2016})}\BibitemShut
  {NoStop}%
\bibitem [{\citenamefont {{Horodecki}}\ and\ \citenamefont
  {{Oppenheim}}(2013)}]{horodecki2013fundamental}%
  \BibitemOpen
  \bibfield  {author} {\bibinfo {author} {\bibfnamefont {M.}~\bibnamefont
  {{Horodecki}}}\ and\ \bibinfo {author} {\bibfnamefont {J.}~\bibnamefont
  {{Oppenheim}}},\ }\href {\doibase 10.1038/ncomms3059} {\bibfield  {journal}
  {\bibinfo  {journal} {Nat. Commun.}\ }\textbf {\bibinfo {volume} {4}},\
  \bibinfo {eid} {2059} (\bibinfo {year} {2013}),\
  10.1038/ncomms3059}\BibitemShut {NoStop}%
\bibitem [{\citenamefont {{Newman}}\ \emph {et~al.}(2016)\citenamefont
  {{Newman}}, \citenamefont {{Mintert}},\ and\ \citenamefont
  {{Nazir}}}]{newman2016performance}%
  \BibitemOpen
  \bibfield  {author} {\bibinfo {author} {\bibfnamefont {D.}~\bibnamefont
  {{Newman}}}, \bibinfo {author} {\bibfnamefont {F.}~\bibnamefont {{Mintert}}},
  \ and\ \bibinfo {author} {\bibfnamefont {A.}~\bibnamefont {{Nazir}}},\
  }\href@noop {} {\bibfield  {journal} {\bibinfo  {journal} {ArXiv e-prints}\ }
  (\bibinfo {year} {2016})},\ \Eprint {http://arxiv.org/abs/1609.04035}
  {arXiv:1609.04035} \BibitemShut {NoStop}%
\bibitem [{\citenamefont {{Korzekwa}}\ \emph {et~al.}(2016)\citenamefont
  {{Korzekwa}}, \citenamefont {{Lostaglio}}, \citenamefont {{Oppenheim}},\ and\
  \citenamefont {{Jennings}}}]{korzekwa2015extraction}%
  \BibitemOpen
  \bibfield  {author} {\bibinfo {author} {\bibfnamefont {K.}~\bibnamefont
  {{Korzekwa}}}, \bibinfo {author} {\bibfnamefont {M.}~\bibnamefont
  {{Lostaglio}}}, \bibinfo {author} {\bibfnamefont {J.}~\bibnamefont
  {{Oppenheim}}}, \ and\ \bibinfo {author} {\bibfnamefont {D.}~\bibnamefont
  {{Jennings}}},\ }\href {http://stacks.iop.org/1367-2630/18/i=2/a=023045}
  {\bibfield  {journal} {\bibinfo  {journal} {New Journal of Physics}\ }\textbf
  {\bibinfo {volume} {18}},\ \bibinfo {pages} {023045} (\bibinfo {year}
  {2016})}\BibitemShut {NoStop}%
\end{thebibliography}%

\appendix

\section{Corrections to Landauer erasure due to discreteness}
\label{appendix_discrete}

Following \cite{vaccaro2011information} (Section~(3b)), one finds that the angular momentum cost of erasure using a spin bath is 
\be
\Delta Q = \sum_{n=1}^\infty \hbar \frac{e^{-\alpha n}}{1+e^{-n \alpha}}.
\ee
Note that, from a more general perspective, this is analogous to the standard formula for the erasure cost used in Sec.~\ref{sec:qubitprotocol}, with the only difference that a discreteness in the operation of level transformation has been introduced. By an appropriate reinterpretation of the constant $\hbar$, hence, the previous formula not only describes processes in which a fundamental discreteness prevents a continuous change of the level structure, but also experimental limitations acting to the same effect.
 
The above sum can be expressed in terms of elementary functions, the $q$-digamma functions $\mathcal{\psi}_{q}(z)$,
\be
\label{eq:discrete}
\Delta Q = \frac{\hbar}{2} + \frac{\hbar}{\alpha}[i \pi + \log(e^\alpha -1) + \psi_{e^\alpha}(-i\pi/\alpha)].
\ee
One can check that in the infinite temperature limit $\Delta Q \approx \alpha^{-1}\log(2)$, as expected in the continuous case, whereas for temperatures close to zero $\Delta Q$ approaches $\hbar/2$. The reason is, of course, quantisation of angular momentum and the fact that erasure requires at least one step, whose cost is $\hbar /2$ (see Figure~\ref{fig:discrete}).

\begin{figure}[h!]
	\includegraphics[width=7cm]{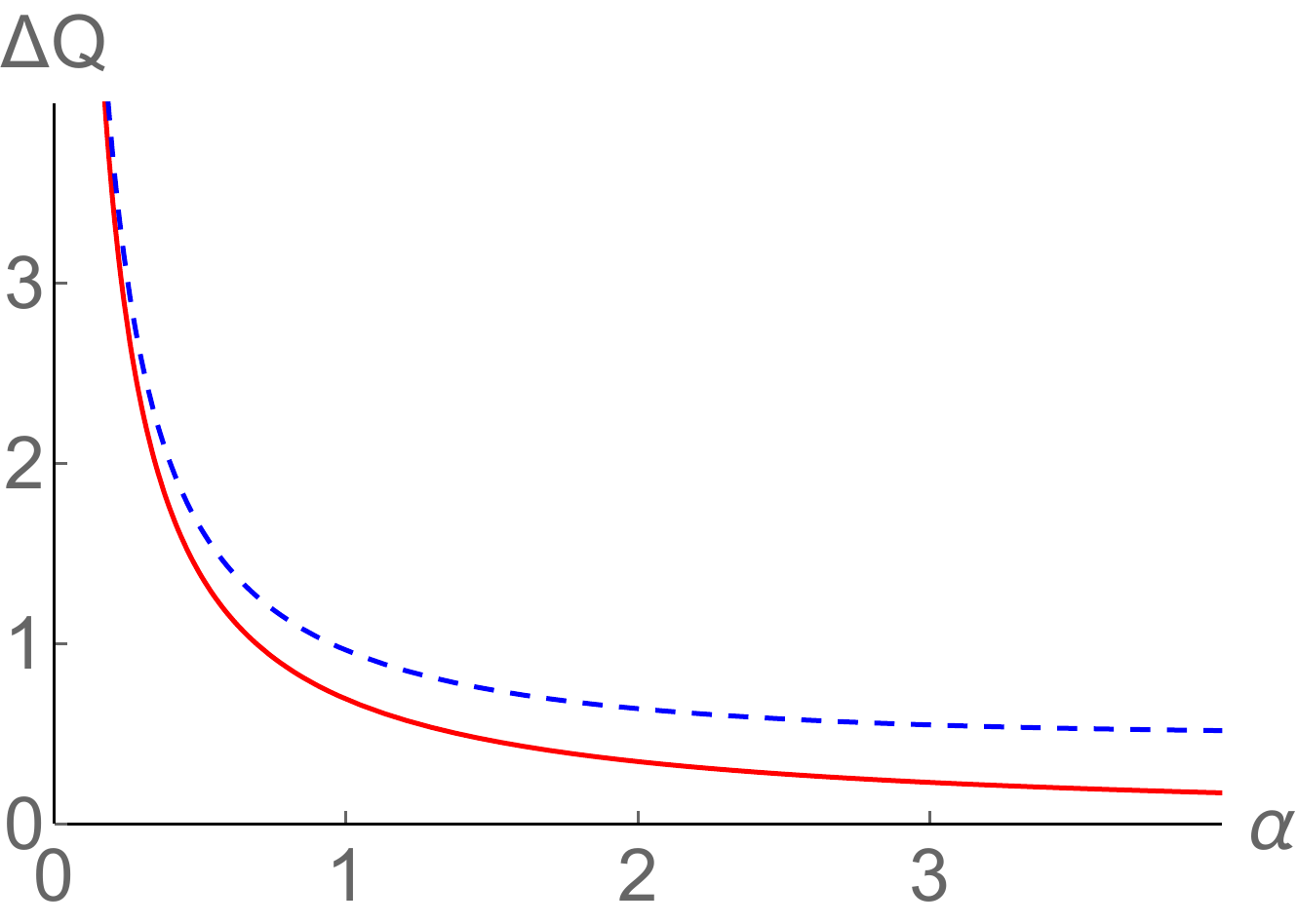}
	\caption{\label{discreteerasure} \textbf{Angular momentum baths and $\hbar$ discreteness.} The red continuous curve is the Landauer bound for a single charge in the continuous case, $\Delta Q = \alpha^{-1} \ln 2$. The blue dashed curve represents the erasure cost $\Delta Q$ as a function of $\alpha$ when level raising is discrete (as given by Eq.~\eqref{eq:discrete}, where we set units $\hbar=1$). Discreteness is relevant at low temperatures; in particular the minimum erasure cost equals the minimum allowed change in the level structure and not $0$ for $\alpha \rightarrow \infty$, since erasure requires at least one step to be performed.}
	\label{fig:discrete}
\end{figure}

	\section{Properties of the maximum entropy projection}
	\label{appendix_ME}
	
	Once $\mathcal{C}$ is fixed, consider the mapping defined in the main text:
	\be
	\rho \mapsto \rho_{\mathcal{C}}.
	\ee
From an inference point of view, $\rho_{\mathcal{C}}$ is the most unbiased state among all those compatible with the observed data, when uncertainty is measured through the von Neumann entropy. This mapping is obviously positive. However we now show that, in general, it is not completely positive.
	
	Let $\mathcal{C} =\{ \I, X  , Y \}$, where $X$ and $Y$ are Pauli operators in an effective qubit subsystem. Aside from the energy spectrum constraints (which are independent due to commutativity), $\rho_{\mathcal{C}}$ takes the form $\rho_{\mathcal{C}} \propto e^{ -\alpha_x X - \alpha_y Y}$. The maximum entropy projection $\rho \stackrel{\mathcal{E}}{\mapsto} \rho_{\mathcal{C}}$ selects $\alpha_x$ and $\alpha_y$ such that $\tr{}{\rho X} = \tr{}{\rho_{\mathcal{C}}X}$ and $\tr{}{\rho Y} = \tr{}{\rho_{\mathcal{C}}Y}$. For this particular model, such a map $\E$ turns out to be the well-known ``pancake map'', whose image when acted on the whole Bloch sphere is the $x$-$y$ equatorial plane. This defines a positive, but not \emph{completely positive} quantum map, in the sense that if applied to a subsystem of a maximally entangled state it generates negative probabilities. This is seen by computing $\tilde{\rho}=\E \otimes id ( \ketbra{\Omega}{\Omega})$ on the pure bipartite, maximally entangled state given by $\ket{\Omega}= (\ket{00} + \ket{11})/\sqrt{2}$. Since $\ketbra{\Omega}{\Omega} $ can be written as $\frac{1}{4} ( \I + XX)(\I + ZZ ) = \frac{1}{4} ( \I + XX + ZZ -YY)$, one finds that $\tilde{\rho}= \frac{1}{4} ( \I + XX - YY) $ which has eigenvalues $\{\frac{3}{4}, \frac{1}{4}, \frac{1}{4}, -\frac{1}{4}\}$, and therefore does not correspond to a physically allowed quantum state on the global system. 

	 As expected, the maximum entropy map is an inference procedure, not a physical map. In fact, one can easily check that not even close approximations of the pancake map can be completely positive. For example, an approximation of the pancake map is one in which the full Bloch sphere is shrank to a circle of radius $r$ within the $x$-$y$ equatorial plane; moreover, we can also allow the projection to be approximate, by setting the average value of $z$ to be $\epsilon$ instead of zero. For this approximate pancake map to be completely positive it is necessary that  $r \leq (1+\epsilon)/2$. This shows that not even this approximate mapping is physical.
	
	 	\section{Maximum entropy state is C(\v{\mu})-completely passivity}
	 	\label{appendix_cmucompletepassivity}
	 	
	 	The proof is based on standard arguments   \cite{uzdin2015coherence}.
	 	Consider $S(\rho || \rho_{\mathcal{C}})$, where $\rho_{\mathcal{C}}$ is given by Eq.~\eqref{eq:generalisedgibbs}. Using the expression for the quantum relative entropy, by direct computation one finds $S(\rho || \rho_{\mathcal{C}})=F_{\mathcal{C}}(\rho) - F_{\mathcal{C}}(\rho_{\mathcal{C}}) $, where $F_{\mathcal{C}}(\rho) := \sum_i \mu_i \tr{}{\rho C_i} - S(\rho)$. Since $S(\rho||\rho_{\mathcal{C}}) \ge 0$, with equality if and only if $\rho =\rho_{\mathcal{C}}$, we recover that $\rho_{\mathcal{C}}$ uniquely minimises $F_{\mathcal{C}}$. Now consider any two states $\rho$ and $\sigma$ linked by a unitary transformation, i.e., $\sigma = U \rho U^\dagger$. Then $F_{\mathcal{C}}(\sigma) - F_{\mathcal{C}}(\rho)=  \tr{}{\sigma C(\v{\mu})} - \tr{}{\rho C(\v{\mu})}$, where we used the unitary invariance of the von Neumann entropy. Taking $\rho = \rho_{\mathcal{C}}$, from $F_{\mathcal{C}}(\sigma) \geq F_{\mathcal{C}}(\rho_{\mathcal{C}})$ we obtain that no unitary (more generally, no constant entropy transformation) can decrease the expectation value of $C(\v{\mu})$. Hence, $\rho_{\mathcal{C}}$ is $C(\v{\mu})$-passive. Moreover, since $F_{\mathcal{C}}$ is an \emph{additive} function, precisely the same argument holds for arbitrarily many copies $\rho_{\mathcal{C}}^{\otimes n}$ and thus one deduces complete passivity in the same manner.

\section{Conservation law within the resource theory framework}
\label{appendix_assumption}

The assumptions of Thermal Operations are extremely minimal, and can be phrased informally as
\begin{enumerate}
\item ``Energy is microscopically conserved".
\item ``The Gibbs state is special".
\end{enumerate}
More precisely, we have that all Thermal Operations take the form 
\begin{equation}\label{thermal}
\E(\rho) = \tr{R'}{U (\rho \otimes \gamma_R) U^\dagger}, 
\end{equation}
with $U$ obeying $[U, H\otimes \I_R + \I\otimes H_R]= 0$ expressing that energy is conserved microscopically and \mbox{$\gamma_R = e^{-\beta H_R}/\tr{}{e^{-\beta H_R}}$} for an arbitrary $H_R$ but fixed $\beta$. Here $\tr{R'}{\cdot}$ denotes a partial trace on some of the degrees of freedom $R'$ (in general $R' \neq R$). If the final Hamiltonian of the system is different from the initial one, $H$ includes the degrees of freedom of a clock (see \cite{brandao2011resource}, Appendix H, \cite{horodecki2013fundamental}, Supplementary Note 1, \cite{brandao2013second} Appendix~I).

There are key reasons for making the assumption of microscopic energy conservation, particularly if one wishes to address thermodynamics in extreme quantum regimes. In traditional, classical, macroscopic thermodynamics the energy used is the internal energy of the system, expressed typically in terms of the expectation value of a Hamiltonian, $\tr{}{H \rho}$. However this method of quantifying energy is not made within the resource framework for important reasons, which we now discuss.

An averaged energy condition works well in the classical regime, but in extreme quantum regimes admits transformations that are highly problematic and questionable. For example, the condition allows unbounded injection of coherence into thermal states with no accounting of the back-reaction on the field beyond semi-classical approximations.  More explicitly, consider a system with energy eigenstates $\{|E_k\>\}$, initially described by a thermal Gibbs state $\gamma = \exp[ -\beta H]/Z$, on which we apply the quantum operation $\E$ for which
\begin{align}
\E (\gamma) &= |\psi\>\<\psi|, \nonumber\\
|\psi\> &= \frac{1}{\sqrt{Z}} \sum_k \sqrt{e^{-\beta E_k }} |E_k\>.
\end{align}
Any conservation law based on $\<H\>$, or indeed \emph{any function of the the moments} $\{ \<H^k\>\}_{k \in \mathbb{N}}$, will never forbid such a transformation, despite transforming to a pure quantum state and generating coherence for free. Even if $\mathcal{E}$ is required to be a unitary process, injection of quantum coherence is still possible, e.g. if $H = \ketbra{1}{1} + 2 \ketbra{2}{2}$, $\ket{1} \rightarrow (\ket{0} + \ket{2})/\sqrt{2}$ by a unitary $U$ that preserves the average energy. Therefore in regimes in which coherence properties are significant, an averaged energy condition is highly questionable and so we make the stronger assumption of microscopic energy conservation.

Given the form of Thermal Operations, it is also important to emphasize that these transformations:
\begin{enumerate}
\item Do not assume that the experimenter has \emph{microscopic control} over the unitary $U$ that is implemented. 
\item Do not assume weak-coupling, or that system \emph{couplings are switched off} at the end.
\item Do not clash with the application of \emph{time-dependent Hamiltonians} in thermodynamics.
\end{enumerate}
Thermal operations only require that the operations applied each have a Stinespring dilation of the form (\ref{thermal}). Operationally it is $\E$ that the experimenter can apply to the system and not the unitary $U$, which is anyway not unique. In fact, a large number of unitaries will induce the same transition, as it is discussed explicitly for the case of work distillation in \cite{brandao2011resource}, Appendix H. 

Point (2) is less obvious in the context of the decoupling of the system from a thermal bath, or the presence of strong couplings. Thermal Operations permit arbitrarily strong couplings, and \emph{do not} assume a perfect decoupling at the end of the process. The latter is permitted by the fact that Eq.~(\ref{thermal}) involves $\tr{R'}{\cdot}$ and not $\tr{R}{\cdot}$. In other words, the working system need not remain the same throughout and so is compatible with strong-coupling techniques such as reaction core methods~\cite{newman2016performance}.

The assumption of arbitrary time-dependent Hamiltonians is ultimately an asymptotic approximation and makes developing an information-theoretic theory of quantum coherence a subtle one. The framework of Thermal Operations provides a rigorous way to deal with such matters. Namely, one makes explicit how the time-dependent Hamiltonian is actually realised through a quantum reference frame. We therefore have thermodynamic processes of the form
\begin{equation}
\rho \otimes \chi \rightarrow \tr{R'}{U ( \rho \otimes \chi \otimes \gamma_R )U^\dagger}
\end{equation}
where $\chi$ is a quantum reference frame state that acts as a clock. When the reference frame is unbounded in the amount of coherence it has, then one can recover perfect control over time-dependent Hamiltonians \cite{bartlett2007reference}. However, for regimes in which coherence matters, one makes explicit back-actions and the irreversible depletion of quantum coherence \cite{korzekwa2015extraction}, which cannot be accounted for under the assumption of time-dependent Hamiltonians.

\end{document}